\documentclass[12pt]{article}

\usepackage[T1]{fontenc}              
\usepackage[width=16cm, height=22cm]{geometry}
\usepackage{amsmath,amssymb,bm,color,mathrsfs,mathtools}
\usepackage[amsmath,thmmarks,hyperref]{ntheorem} 
\usepackage{amscd,verbatim}
\usepackage{authblk}

\newcommand{\N}{{\mathbb N}}
\newcommand{\R}{{\mathbb R}}
\newcommand{\T}{{\mathbb T}}
\newcommand{\Z}{{\mathbb Z}}

\theoremstyle{nonumberplain}  
\theoremheaderfont{\itshape}  
\theorembodyfont{\normalfont}  
\theoremseparator{.}  
\theoremsymbol{$\Box$}
\newtheorem{proof}{Proof} 
  
\theoremstyle{plain}  
\theoremheaderfont{\normalfont\bfseries}  
\theorembodyfont{\itshape}  
\theoremsymbol{~}
\theoremseparator{.}  
\newtheorem{proposition}{Proposition}[section]  

\newtheorem{lemma}[proposition]{Lemma}  
\newtheorem{theorem}[proposition]{Theorem}   

\theorembodyfont{\normalfont}  
\newtheorem{remark}[proposition]{Remark}

\newtheorem{definition}[proposition]{Definition} 

\theoremstyle{nonumberplain}
\theoremheaderfont{\normalfont\bfseries}  
\theorembodyfont{\itshape}  
\theoremsymbol{~}
\theoremseparator{.}  
\newtheorem{theoremnonumber}[proposition]{Main Theorem}  
\newtheorem{mainquestion}[proposition]{Main Question}

\begin{document}

\title{Large-scale geometry obstructs localization}

\author[1]{Matthias Ludewig\thanks{matthias.ludewig@mathematik.uni-regensburg.de}}
\author[2]{Guo Chuan Thiang\thanks{guochuanthiang@bicmr.pku.edu.cn}}
\affil[1]{\normalsize Fakult\"{a}t f\"{u}r Mathematik, Universit\"{a}t Regensburg}
\affil[2]{\normalsize Beijing International Center for Mathematical Research, Peking University}

\date{\today}

\maketitle

\begin{abstract}
We explain the coarse geometric origin of the fact that certain spectral subspaces of topological insulator Hamiltonians are delocalized, in the sense that they cannot admit an orthonormal basis of localized wavefunctions, with respect to any uniformly discrete set of localization centers. This is a robust result requiring neither spatial homogeneity nor symmetries, and applies to Landau levels of disordered quantum Hall systems on general Riemannian manifolds.
  \end{abstract}

\section{Introduction}
Schr\"{o}dinger operators on Euclidean space $\R^d$ with a lattice-periodic potential typically have absolutely continuous spectrum, comprising a sequence of possibly overlapping bands $S_i=[a_i,b_i]$. Although there are no normalizable eigenfunctions, the spectral subspace of an isolated band $S_i$ is an invariant subspace for the lattice $\Z^d$ of translations, and does admit an orthonormal basis $\{\gamma^*w\}_{\gamma\in\Z^d}$ comprising the $\Z^d$-translates of some reference \emph{Wannier wavefunction} $w\in L^2(\R^d)$. Such a basis is called a \emph{periodic Wannier basis} for the spectral subspace, and is a fundamental tool in solid-state physics.

The choice of $w$ is highly non-unique, and usually a \emph{localization} criterion with respect to a localization center $x_0\in\R^d$ is requested of it. Then the $\Z^d$-orbit of $x_0$ is the set of localization centers for the Wannier basis functions. Remarkably, in $d=2$, exponential (and also much milder) localization is obstructed \cite{Brouder} by the first Chern class of the eigenbundle over the Pontrjagin dual $\T^2$ obtained via Bloch--Floquet transform \cite{Kuchment}. Therefore, there exist spectral subspaces of periodic Schr\"{o}dinger operators which are not periodic-Wannier-localizable. Concrete examples include Landau level eigenspaces of magnetic Laplacians \cite{AvSS,BES,Kunz,LT-wannier} and discrete model analogues such as Chern insulators \cite{Haldane}.

In \cite{LT-wannier}, we initiated the study of Wannier bases $\{\gamma^*w\}_{\gamma\in G}$ for (magnetic) Schr\"{o}dinger operators on general Riemannian manifolds $X$, invariant under \emph{non-abelian} discrete groups $G$ of isometries. We found that $G$-periodic-Wannier-localizability is obstructed by an element of the $K$-theory of $C^*_r(G)$, where $C^*_r(G)$ denotes the reduced group $C^*$-algebra of $G$. The special case $C^*_r(\Z^2)\cong C(\T^2)$ reduces to the familiar Chern class obstruction encountered when $X$ is the Euclidean plane $\R^2$.

The $K_0(C^*_r(G))$ obstruction is still deficient in several ways:
\begin{enumerate}
\item[(1)] $G$-invariance is lost when disorder is present, so no $G$-periodic Wannier bases are available, localized or otherwise.
\item[(2)] $X$ may be very inhomogeneous, with no notion of isometric $G$-action at all. For example, a 2D quantum Hall or Chern topological insulator is usually modelled with $X$ a Euclidean plane, but it is robust against non-isometric deformations of the embedding of $X$ within the laboratory Euclidean $\R^3$. Other examples with inhomogeneous metric space $X$ are amorphous topological insulators \cite{amorphous}.
\item[(3)] In quantum Hall systems, the externally applied magnetic field strength is realistically non-uniform, so there are no exact magnetic translational symmetries even if $X$ is assumed to be a simple geometric space like $\R^2$.
\end{enumerate}
These deficiencies prompt us to drop the periodicity condition on Wannier bases entirely, and allow their localization centers to come from  an arbitrary uniformly discrete subset $\Gamma\subset X$. 
We are led to the general notion of a \emph{uniformly localized} Wannier basis, made precise in Definition~\ref{DefinitionUniformlyLocalized} below.
For $X$ a Euclidean space, a similar relaxation has been considered in \cite{MMP,NN}. 

\begin{mainquestion}
 Does a given subspace $\mathcal{H} \subseteq L^2(X)$  admit a uniformly localized Wannier basis, whose localization centers come from some uniformly discrete $\Gamma\subset X$? If not, what is the obstruction?
\end{mainquestion}

Our main result (Theorem \ref{thm:main} in the main text) works under certain geometric conditions on $X$, which are in particular satisfied for $X = \R^d$.
The result states that if the orthogonal projection onto $\mathcal{H}$ lies in the Roe $C^*$-algebra $C^*(X)$, then the existence of a uniformly localized Wannier basis is obstructed by the corresponding $K$-theory class:

\begin{theoremnonumber}
Let $\mathcal{H} \subseteq L^2(X)$ be a subspace such that the corresponding projection defines a non-trivial element in the $K$-theory of the Roe algebra $C^*(X)$.
Then $\mathcal{H}$ does not admit a uniformly localized Wannier basis with localization centers $\Gamma\subset X$, for any choice of uniformly discrete $\Gamma\subset X$ whatsoever.
\end{theoremnonumber}

In Section \ref{sec:Landau}, we prove that the Landau bands of disordered Landau Hamiltonians encounter this obstruction, so they are not Wannier-localizable in any reasonable sense.

The group $K_0(C^*(X))$ is a computable \emph{coarse}, or \emph{large-scale}, geometric invariant of $X$, and makes sense for a large class of metric spaces. Also, a large class of spectral projections of geometric operators on $X$ belong to $C^*(X)$, and this has been exploited in \cite{EwertMeyer,KLT, LT-hyp}. Of particular conceptual importance is the fact that $X$ need not be rigidly fixed --- small-scale geometric deformations and even topology changes (such as puncturing holes in the sample) are allowed. 
Therefore, the property of non-Wannier-localizability is extremely robust and widely applicable, and is justified as a possible organizing principle for the study of \emph{topological phases} in physics.

\section{Projections in Roe algebras}

Let $X$ be a Riemannian manifold with distance function $d(\cdot,\cdot)$. The Hilbert space $L^2(X)$ is defined with respect to the measure induced from the Riemannian volume element. The algebra $C_0(X)$ is represented on $L^2(X)$ by pointwise-multiplication operators. 

The following two definitions may be found in \cite{EwertMeyer, HRY,  Roebook}.

\begin{definition}\label{dfn:propagation.definitions}
A bounded operator $T\in \mathcal{B}(L^2(X))$ is \emph{locally compact} if the operators $Tf, fT$ are compact whenever $f$ has compact support. It has \emph{finite propagation} if there exists some $R>0$ such that $fTg=0$ whenever the supports of $f,g\in C_0(X)$ are a distance $R$ or more apart. It is \emph{supported near} a subset $Y\subset X$ if there exists $R> 0$ such that $fT=0=Tf$ whenever $f\in C_0(X)$ has $d({\rm supp}\,f, Y)\geq R$.
\end{definition}

\begin{definition}
The \emph{Roe $C^*$-algebra} $C^*(X)$ is the norm-completion of the $*$-algebra of locally compact, finite propagation operators on $L^2(X)$.
For a closed subset $Y\subset X$, the \emph{Roe $C^*$-algebra localized near $Y$}, denoted $C^*_X(Y)$, is the closed ideal in $C^*(X)$ generated by those $T\in C^*(X)$ supported near $Y$.
\end{definition}

A subset $\Gamma\subset X$ is \emph{uniformly discrete} if there exists $r>0$ such that the open balls $B_r(\gamma), \gamma\in\Gamma$ are disjoint. The largest such $r$ is its \emph{packing radius}.

\begin{definition}
Let $\Gamma$ be a uniformly discrete subset of $X$, and $\mathcal{H}$ a Hilbert subspace of $L^2(X)$. A \emph{$\Gamma$-compactly supported Wannier basis for $\mathcal{H}$} is an orthonormal basis $\{w_\gamma\}_{\gamma\in\Gamma}$ for $\mathcal{H}$, such that ${\rm supp}(w_\gamma)\subset B_r(\gamma)$ for all $\gamma\in\Gamma$, where $r$ is the packing radius of $\Gamma$.
\end{definition}

A $\Gamma$-compactly supported Wannier basis is a ``classical'' decomposition of $\mathcal{H}$, in the sense that the wavefunctions $w_\gamma$ are ``pointlike'' with no spatial overlaps whatsoever. 

\begin{lemma}\label{lem:compact.Wannier.in.Roe}
Suppose $\mathcal{H}\subset L^2(X)$ admits a $\Gamma$-compactly supported Wannier basis for some uniformly discrete $\Gamma\subset X$. 
Then the orthogonal projection $p=p_\mathcal{H}$ onto $\mathcal{H}$ belongs to $C^*(X)$.
\end{lemma}
\begin{proof}
For $f \in C_0(X)$, both $pf$ and $fp$ are of finite rank, hence compact.
Moreover, $p$ has propagation at most the packing radius $r$, by assumption.
\end{proof}

Any projection $p$ in the Roe algebra determines an element $[p]$ in the $K$-theory group $K_0(C^*(X))$. 
The following proposition shows that if such a projection is the orthogonal projection onto a subspace $\mathcal{H} \subset L^2(X)$ admitting a $\Gamma$-compactly supported Wannier basis, then its $K$-theory class is automatically trivial for a large class of spaces $X$.

\begin{proposition}\label{prop:vanishing}
Suppose $X$ admits a decomposition $X=X_1\cup X_2$ into closed subspaces $X_i$, such that $K_0(C^*(X_i))=0, i=1,2$. 
Suppose the range $\mathcal{H}$ of a projection $p\in C^*(X)$ admits a $\Gamma$-compactly supported Wannier basis for some uniformly discrete $\Gamma\subset X$. Then $[p]=0$ in $K_0(C^*(X))$.
\end{proposition}

\begin{proof}
Set $\Gamma_1=X_1\cap\Gamma$. 
Given a $\Gamma$-compactly supported Wannier basis $\{w_\gamma\}_{\gamma \in \Gamma}$  for $\mathcal{H}$, set
\begin{equation*}
\mathcal{H}_1 =\overline{{\rm span}\{w_\gamma\,:\,\gamma\in \Gamma_1\}}, \qquad
\mathcal{H}_2 =\overline{{\rm span}\{w_\gamma\,:\,\gamma\in \Gamma\setminus\Gamma_1\}}.
\end{equation*}
Let $p_i$ be the respective orthogonal projections onto $\mathcal{H}_i, i=1,2$.
By Lemma \ref{lem:compact.Wannier.in.Roe}, we have $p=p_\mathcal{H}\in C^*(X)$, while $p_i$ lies in the Roe algebra $C^*_X(X_i)$ localized near $X_i$. 
 
For $i=1,2$, write $j_i:C^*_X(X_i)\to C^*(X)$ for the inclusion map.
By construction, we have an orthogonal sum
\begin{equation*}
p=p_\mathcal{H}=j_1(p_1)+j_2(p_2)\;\in C^*(X).\label{eqn:basic.split}
\end{equation*}
Because of the isomorphism $K_0(C^*_X(X_i))\cong K_0(C^*(X_i))$ (see \S5, Lemma 1 in \cite{HRY}), we have $[p_i]\in K_0(C^*_X(X_i)) \cong K_0(C^*(X_i))=0$ by assumption. Hence
\begin{equation*}
[p]=[j_1(p_1)]+[j_2(p_2)]=(j_1)_*[p_1]+(j_2)_*[p_2]=0\in K_0(C^*(X)).
\end{equation*}
\end{proof}

\begin{remark}
A standard example satisfying the assumptions of Prop.~\ref{prop:vanishing} is $X$ a Euclidean space, with $X_1, X_2$ the upper and lower Euclidean half-spaces. But there are many more examples, and in fact, the results of this section may be applied to the general coarse geometric setting where $X$ is a proper metric space, and $L^2(X)$ is replaced by an \emph{ample $X$-module}, i.e., a Hilbert space $\mathscr{H}$ equipped with a non-degenerate $*$-representation of $C_0(X)$ such that no non-zero $f\in C_0(X)$ acts as a compact operator. 
In later sections, we will deal with differential operators on $X$, so we eschew such generality.
\end{remark}

\section{Non-existence of localized Wannier bases}

From now on, let $X$ be a Riemannian manifold of dimension $d$.
We assume that $X$ has \emph{polynomial volume growth} in the sense that there exist constants $A, \nu >0$ such that for all $x \in X$ and $R\geq 0$, 
\begin{equation} \label{PolynomialGrowthCondition}
  \mathrm{vol}(B_R(x)) \leq A (1+R)^\nu. 
\end{equation}
We moreover assume that the volume of small balls is bounded from below, i.e., there exist $v, r_0 >0$ such that for all $0 \leq \rho \leq r_0$ and $x\in X$,
\begin{equation} \label{BallEstimateBelow}
  \mathrm{vol}(B_{\rho}(x)) \geq v \rho^d.
\end{equation}

As discussed below, we are interested in the existence of Wannier bases for spectral subspaces associated to Schr\"odinger operators.
For purely analytic reasons, in most cases of interest, there is no hope of finding compactly supported Wannier basis functions $w_\gamma$, so we loosen the localization criterion as follows.

\begin{definition} \label{DefinitionUniformlyLocalized}
Let $\Gamma\subset X$ be a uniformly discrete subset. 
A set of normalized functions $\{w_\gamma\}_{\gamma\in\Gamma}\subset L^2(X)$ is \emph{$\Gamma$-uniformly localized} if for all $\mu>0$, there exists a constant $C_\mu>0$ such that
\begin{equation}
|w_\gamma(x)|\leq C_\mu (1 + d(x, \gamma))^{-\mu},\qquad \forall\gamma\in\Gamma, x \in X.\label{eqn:UL.decay}
\end{equation}
A $\Gamma$-uniformly localized orthonormal basis for a Hilbert subspace $\mathcal{H}\subset L^2(X)$ is called a \emph{$\Gamma$-uniformly localized Wannier basis} for $\mathcal{H}$.
\end{definition}

\begin{remark}\label{rem:subdivision}
In \cite{MMP}, each localization center $\gamma\in\Gamma$ is allowed to have a degeneracy $m(\gamma)\geq 1$, uniformly bounded by some $m_*$. 
By enlarging the set $\Gamma$ of centers, e.g.\ by replacing each $\gamma$ with $m_*$ distinct points in $B_r(\gamma)$, we may remove the degeneracy at the expense of reducing the packing radius.
This is possible while maintaining a positive packing radius because the numbers $m(\gamma)$ are uniformly bounded by $m_*$.
\end{remark}

\begin{remark}
With the assumption that $X$ has polynomial growth, the condition of \eqref{eqn:UL.decay} implies the $L^2$-condition
\begin{equation*}
\int_{X} |w_\gamma(x)|^2(1+d(x,\gamma))^{\mu}\;\leq \tilde{C}_\mu,\qquad \forall\gamma\in\Gamma,
\end{equation*}
for any $\mu \geq 0$.
$L^2$-decay conditions are sometimes preferred in studies of Wannier bases. For $X=\R^d$, the projection for a spectral subspace $\mathcal{H}$ is often known to have exponentially-decaying kernel away from the diagonal (Combes--Thomas estimates). In such situations, if $\mathcal{H}$ admits a $\Gamma$-uniformly localized Wannier basis with decay condition in the above $L^2$ sense, then it also admits one with the $L^\infty$ decay condition; see,  e.g., Lemma 2.6 of \cite{MMP}.
\end{remark}

\begin{proposition}\label{prop:vN.argument}
Suppose a Hilbert subspace $\mathcal{H}\subset L^2(X)$ admits a $\Gamma$-uniformly localized Wannier basis. 
Then its spectral projection $p_\mathcal{H}$ is contained in $C^*(X)$ and, as an element of this algebra, is Murray--von Neumann equivalent to a projection onto the span of some $\Gamma$-compactly supported orthonormal set.
\end{proposition}

We need the following lemma.

\begin{lemma} \label{LemmaEstimateContDisc}
Let $\nu$ be the volume growth exponent of $X$ from \eqref{PolynomialGrowthCondition}.
Then for all $\mu > \nu$, there exists $C_1, C_2>0$ such that for all $x\in X$ and $R\geq 0$,
\begin{align}
 \label{PolyGrowthCont} 		
 \int_{X \setminus B_R(x)} (1+d(x, y))^{-\mu} dy &\leq C_1 (1+R)^{\nu-\mu}\\
\text{and} \qquad \qquad \sum_{\substack{\gamma \in \Gamma \\ d(x, \gamma) \geq R}} (1+d(x, \gamma))^{-\mu} &\leq C_2 (1+R)^{\nu-\mu}.
 \label{PolyGrowthDisc}
\end{align}
\end{lemma}

\begin{proof}
By the coarea formula, we get
\begin{equation*}
\int_{X \setminus B_R(x)} (1+d(x, y))^{-\mu} dy = \int_R^\infty (1+s)^{-\mu} \frac{d}{d s} \mathrm{vol}(B_{s}(x)) ds  \leq \frac{A\nu}{\mu-\nu}(1+R)^{\nu - \mu},
\end{equation*}
where in the second step, we integrated by parts, used \eqref{PolynomialGrowthCondition} and evaluated the integral.
This gives the estimate \eqref{PolyGrowthCont} with $C_1 = \frac{A\mu}{\mu-\nu}$.

The second estimate follows from the first. 
Indeed, let $0<\varepsilon<\min\{1,r\}$ and $R\geq 0$.
Then for any $\gamma \in \Gamma$ with $d(x, \gamma) \geq R$ and any $y\in B_\varepsilon(\gamma)$, we have
\begin{equation}
d(x,y) >R-\varepsilon,\label{eqn:distance1}
\qquad \mathrm{and} \qquad
1+d(x,\gamma)>1+d(x,y)-\varepsilon>0.
\end{equation}
With these, we can estimate the following sum for any $x\in X$,
\begin{align*}
  \sum_{\substack{\gamma \in \Gamma \\ d(x, \gamma) \geq R}} (1+d(x, \gamma))^{-\mu}
  &=\sum_{\substack{\gamma \in \Gamma \\ d(x, \gamma) \geq R}} \frac{1}{\mathrm{vol}(B_{\varepsilon}(\gamma))}\int_{B_\varepsilon(\gamma)}\left(1+d(x, \gamma)\right)^{-\mu} dy\\
  &\leq  \sum_{\substack{\gamma \in \Gamma \\ d(x, \gamma) \geq R}} \frac{1}{\mathrm{vol}(B_{\varepsilon}(\gamma))}\int_{B_\varepsilon(\gamma)}\left(1+d(x, y)-\varepsilon\right)^{-\mu} dy & &\mathrm{by\;}\eqref{eqn:distance1}\\
  &\leq \frac{1}{v \varepsilon^d} \sum_{\substack{\gamma \in \Gamma \\ d(x, \gamma) \geq R}}\int_{B_\varepsilon(\gamma)} \left(1 + d(x, y)- \varepsilon\right)^{-\mu} dy& &\mathrm{by\;}\eqref{BallEstimateBelow}\\
  &\leq \frac{1}{v \varepsilon^d}\int_{X\setminus B_{R-\varepsilon}(x)}\left(1 + d(x, y)- \varepsilon\right)^{-\mu} dy & &\mathrm{by\;}\eqref{eqn:distance1}\\
  &= \frac{1}{v \varepsilon^d}\int_{X\setminus B_{R-\varepsilon}(x)}\left(1 + d(x, y)\right)^{-\mu}\left(1-\frac{\varepsilon}{1+d(x,y)}\right)^{-\mu} dy\\
  &\leq \frac{(1-\varepsilon)^{-\mu}}{v\varepsilon^d}\int_{X\setminus B_{R-\varepsilon}(x)}(1+d(x,y))^{-\mu}\,dy\\
  &\leq  \frac{(1-\varepsilon)^{-\mu}C_1}{v \varepsilon^d} (1+R-\varepsilon)^{\nu-\mu}.& &\mathrm{by\;}\eqref{PolyGrowthCont}
\end{align*}
We may replace $(1+R-\varepsilon)^{\nu-\mu}$ by $(1+R)^{\nu-\mu}$ at the expense of further increasing the constant, thus arriving at the estimate \eqref{PolyGrowthDisc}.
\end{proof}

\begin{proof}[of Prop.~\ref{prop:vN.argument}]
Let $\{w_\gamma\}_{\gamma \in \Gamma}$ be a $\Gamma$-uniformly localized Wannier basis for $\mathcal{H}$. 
As usual, $r$ denotes the packing radius of $\Gamma$, while $r_0$ is the constant in \eqref{BallEstimateBelow} controlling the volume of small balls in $X$.
Choose a $\Gamma$-compactly supported orthonormal set $\{v_\gamma\}_{\gamma\in\Gamma}$ such that each $v_\gamma$ is supported in $B_{\rho}(\gamma)$, where $\rho<\min\{r_0,r\}$. 
Due to \eqref{BallEstimateBelow}, we may assume that the $v_\gamma$ have been chosen such that a bound $|v_\gamma(x)|\leq c_0$ holds for all $\gamma\in\Gamma, x\in X$.

Set
\begin{equation} \label{FormulaForV}
  V = \sum_{\gamma \in \Gamma} w_\gamma \otimes v_\gamma^*,
\end{equation}
where the sum converges strongly.
Then $V$ is a Murray--von Neumann equivalence between the orthogonal projection $p_\mathcal{H}$ onto $\mathcal{H}$, and the projection onto the span of the $v_\gamma$.

We claim that $V$ is contained in $C^*(X)$.
To this end, for $R>0$, let $w_\gamma^R$ be the function that coincides with $w_\gamma$ on $B_R(\gamma)$ and vanishes on $X \setminus B_R(\gamma)$.
Let $V^R$ be the operator defined by the same formula \eqref{FormulaForV} but using $w_\gamma^R$ instead of $w_\gamma$.
Clearly, $V^R$ is locally finite rank and has propagation at most $\max\{R,r\}$, hence $V^R \in C^*(X)$.

We need to estimate the operator norm of $V-V^R$.
Write $K^R(x, y)$ for the kernel of $V-V^R$. 
Fix a $\mu > \nu$, where $\nu$ is the volume growth exponent of $X$ from \eqref{PolynomialGrowthCondition}. For those $y\in X$ that lie inside $B_r(\gamma)$ for some $\gamma\in\Gamma$, we calculate
\begin{align*}
  \int_X |K^R(x, y)|\, dx\; &= \int_X |w_\gamma(x) - w_\gamma^R(x)||v_\gamma(y)|\, dx \\
    &\leq |v_\gamma(y)| \cdot \int_{X \setminus B_R(\gamma)} (1+d(x, \gamma))^{-\mu}\, dx \\
    &\leq c_0 C_1 (1+R)^{\nu-\mu}  =: \delta_1(R).& &\mathrm{by\;}\eqref{PolyGrowthCont} \\
\end{align*}
As for those $y\in X$ which are not contained in any $B_r(\gamma)$, the above inequality still holds since the integrand vanishes.

Next, by ${\rm supp}(v_\gamma)\subset B_\rho(\gamma)\subset B_r(\gamma)$, the bound $|v_\gamma(x)|\leq c_0$, the H\"{o}lder inequality, and the estimate \eqref{PolynomialGrowthCondition},
\begin{equation*}
\int_X |v_\gamma(y)|dy =\int_{B_r(\gamma)}|v_\gamma(y)|dy\leq  c_0\cdot\mathrm{vol}(B_r(\gamma))\leq c_0 A(1+r)^\nu.
\end{equation*}
Then for all $x\in X$, we have the estimate
\begin{align*}
  \int_X |K^R(x, y)| dy &\leq \sum_{\gamma \in \Gamma} \int_X |w_\gamma(x)-w_\gamma^R(x)||v_\gamma(y)| dy\\
  &\leq C_\mu \cdot \sum_{\substack{\gamma \in \Gamma \\ d(x, \gamma) \geq R}} (1+d(x, \gamma))^{-\mu} \cdot \int_X|v_\gamma(y)|dy  & & \mathrm{by\;}\eqref{eqn:UL.decay}\\
  &\leq C_\mu \cdot C_2 (1+ R)^{\nu-\mu} \cdot c_0 A (1+r)^{\nu} =: \delta_2(R). & &\mathrm{by\;}\eqref{PolyGrowthDisc}\\
\end{align*}
For any $f \in L^2(X)$, we therefore obtain
\begin{align*}
  \|(V-V^R)f\|^2 &\leq \int_X\left(\int_X |K^R(x, y)| |f(y)|dy\right)^2 dx \\
  &\leq \int_X\left(\int_X |K^R(x, y)| dy\right) \cdot \left(\int_X |K^R(x, y)| |f(y)|^2 dy\right) dx\\
  &\leq \delta_2(R) \cdot \int_X \left(\int_X |K^R(x, y)| dx\right) |f(y)|^2 dy\\
  &\leq \delta_2(R) \delta_1(R) \|f\|^2.
\end{align*}
Since $\delta_1(R)$ and $\delta_2(R)$ converge to zero as $R\to\infty$, we obtain that $\|V - V^R\|$ converges to zero in this limit.
Thus $V$, as a limit of elements in $C^*(X)$, is also contained in $C^*(X)$.
We conclude that $p_{\mathcal{H}} = VV^* \in C^*(X)$, and $V$ implements the Murray--von-Neumann equivalence between $p_{\mathcal{H}}$ and $V^*V$, the projection onto the subspace spanned by the $v_\gamma$.
\end{proof}

Finally, we arrive at our main theorem, which provides an obstruction to the existence of $\Gamma$-uniformly localized Wannier bases for general Hilbert subspaces $\mathcal{H} \subseteq L^2(X)$.

\begin{theorem}\label{thm:main}
Suppose $X$ admits a decomposition as in Prop.\ \ref{prop:vanishing}. Let $p \in C^*(X)$ be a projection defining a non-trivial element $[p]$ in $K_0(C^*(X))$.
Then its range $\mathcal{H} \subseteq L^2(X)$ does not admit a $\Gamma$-uniformly localized Wannier basis, for any choice of uniformly discrete $\Gamma\subset X$ whatsoever.
\end{theorem}

\begin{proof}
Suppose $\mathcal{H}$ admits a $\Gamma$-uniformly localized Wannier basis. 
Then it follows from Prop.\ \ref{prop:vN.argument} and Prop.\ \ref{prop:vanishing} that $p$ must be trivial in $K_0(C^*(X))$. The contrapositive yields the theorem.
\end{proof}

The main source of projections to be used in Thm.~\ref{thm:main} are spectral projections of magnetic Schr\"{o}dinger operators $H$ on $X$, with smooth and bounded scalar potential and curvature (magnetic field strength) terms. 
We assume $H$ to be self-adjoint and bounded from below, and write $\sigma(H)$ for its spectrum.
Let $S\subset \sigma(H)$ be compact and separated from $\sigma(H)\setminus S$. 
Then its spectral projection $p_S$, with range $\mathcal{H}_S=p_SL^2(X)$, can be defined as $p_S=\varphi_S(H)$, for a suitable continuous function $\varphi_S\in C_0(\R)$. 
As in the case of Dirac operators \cite{Roebook}, we then have \cite{LT-hyp}
\begin{equation*}
p_S=\varphi_S(H)\in C^*(X).
\end{equation*}
Thm.~\ref{thm:main} implies that if $p_S$ is non-trivial in $K_0(C^*(X))$, then $\mathcal{H}_S=p_SL^2(X)$ cannot admit a $\Gamma$-uniformly localized Wannier basis, for any choice of uniformly discrete $\Gamma\subset X$.

\section{Landau band examples and further discussion}\label{sec:Landau}

We will provide examples of $p_S$ which are \emph{non-trivial} in $K_0(C^*(X))$, and which therefore encounters the Wannier-localizability obstruction of Theorem \ref{thm:main}.

\medskip

It is a deep result that for a large class of metric spaces $X$, there is an assembly map \cite{HR-coarse} implementing an isomorphism
\begin{equation*}
\mu_X:K_0(X)\longrightarrow K_0(C^*(X)),
\end{equation*}
where $K_0(X)$ is the Kasparov $K$-homology group. For example, if $X=\R^{2n}$, the isomorphism is $\Z\to \Z$, and is given on a generator by the \emph{coarse index} of the Dirac operator on $\R^{2n}$ \cite{Roebook}. The same is true of the hyperbolic plane and helicoid, and this was exploited in \cite{LT-hyp} and \cite{KLT} respectively.

Therefore, Dirac operator kernels are natural candidate examples of spectral subspaces that are non-Wannier-localizable. However, there is a complication because the coarse index of a Dirac operator 
$
D=\begin{pmatrix} 0 & D^-\\ D^+ & 0 \end{pmatrix}
$
 is not generally represented as
\begin{equation*}
{\rm Ind}(D)=[{\rm ker}D^+]-[{\rm coker}D^+],
\end{equation*}
unless $0$ is known to be isolated in the spectrum of $D$. This spectral separation can be achieved by an appropriate twisting of the standard Dirac operator \cite{LT-hyp,KLT}, corresponding to coupling it to a line bundle with non-vanishing curvature (i.e., a magnetic field).

\medskip

Let us see the effect of the magnetic field on the Dirac spectrum, in the basic example where $X=\R^2$ is the Euclidean plane with Riemannian volume form ${\rm vol}=dx\wedge dy$. For $b>0$, let $D_b$ be the Dirac operator acting on the spinor bundle 
coupled to a line bundle with curvature $b\cdot {\rm vol}$. 
Explicitly, in an appropriate trivialization, and writing $\partial=\partial_x-i\partial_y$, $\bar{\partial}=\partial_x+i\partial_y$, we have
\begin{equation*}
D_b=\begin{pmatrix} 0 & -i\partial +ibx \\-i\bar{\partial}-ibx & 0 	\end{pmatrix},\qquad D_b^2=\begin{pmatrix}
H_b-b & 0 \\ 0 & H_b+b 
\end{pmatrix}\geq 0,
\end{equation*}
where
\begin{equation*}
H_b=-\nabla^2+2ibx\partial_y+b^2x^2=-\partial_x^2-(\partial_y-ibx)^2
\end{equation*}
is the free Landau Hamiltonian (magnetic Laplacian). The two summands in $D_b^2$ have the same non-zero spectrum,
\begin{equation*}
\sigma(H_b-b)\setminus\{0\}=\sigma(\underbrace{H_b+b}_{\geq b\,>\,0})\setminus\{0\}=\sigma(H_b+b).
\end{equation*}
The above shift-invariance in the spectrum of $H_b$ implies that $\sigma(H_b)=(2\N+1)b$, with each eigenvalue $(2n+1)b$ being infinitely-degenerate.

In particular, for $n=0$, we learn that the lowest Landau level (LLL) eigenspace,
\begin{equation*}
{\rm ker}(H_b-b)={\rm ker}(D_b^2)={\rm ker}(D_b)={\rm ker}(D_b^+),
\end{equation*}
is precisely the Dirac kernel, and it is isolated in the Dirac spectrum. Therefore, the LLL spectral projection $p_{\rm LLL}$ for the Landau Hamiltonian $H_b$ is exactly the kernel projection for $D_b$, so
\begin{equation*}
[p_{\rm LLL}]={\rm Index}(D_b)\in K_0(C^*(\R^2)).
\end{equation*}
In turn, this \emph{coarse Dirac index} is a generator of $K_0(C^*(\R^2))\cong\Z$, because the coarse Baum--Connes assembly map is an isomorphism for $X=\R^2$.  

Let us now introduce to the free Landau Hamiltonian $H_b$, a potential function $V$ with ${\sup}_{x}|V(x)|<b$. Then the magnetic Schr\"{o}dinger operator $H_b+V$ will have spectrum lying within a disjoint union of bands,
\begin{equation*}
\sigma(H_b+V)\subset \bigsqcup_{n\in\N} S_n,\qquad S_n=\left(2nb,2(n+1)b\right).
\end{equation*}
We call $\sigma(H_b+V)\cap S_n$ the \emph{$n$-th Landau band} (even though it may not necessarily be a connected interval). Spectral gaps are maintained between adjacent Landau bands, so we may still ask for the projection $p_{n,V}$ onto the spectral subspace for the $n$-th Landau band. For instance, let $\varphi$ be a continuous (or even smooth) bump function with $\varphi(\lambda)=1, \lambda\in S_0$, and vanishing on the other all the other $n>0$ Landau bands. Since $\varphi\in C_0(\R)$, we have
\begin{equation*}
p_{0,V}=\varphi(H_b+V)\in C^*(\R^2).
\end{equation*}
Using the resolvent identity, we observe that $t\mapsto H_b+tV$ is norm-resolvent continuous in $t$, so decreasing $t$ from $1$ to $0$ produces a a homotopy of projections in the Roe algebra,
\begin{equation*}
p_{0,V}=\varphi(H_b+V)\sim_h \varphi(H_b)=p_{\rm LLL}\in C^*(\R^2),
\end{equation*}
and therefore $[p_{0,V}]=[p_{\rm LLL}]\neq 0$ in $K_0(C^*(\R^2))$. 

Invoking Theorem \ref{thm:main}, we conclude that the spectral subspace for $0$-th Landau band is not Wannier localizable, independently of the potential $V$ satisfying ${\rm sup}_x|V(x)|<b$, whatever the choice of uniformly discrete $\Gamma\subset \R^2$ for the localization centers. The same applies to the $n>0$ Landau bands, because each Landau level projection also represents the Dirac index class \cite{LT-hyp}.

\paragraph{Discussion.} 

On the one hand, a disordered Landau Hamiltonian should have some energy interval $I$ for which the spectral subspace is spanned by localized eigenstates --- this allows for invariance of Hall conductance as the Fermi energy is varied within $I$ (the plateaux in the famous experiments). For this effect, a weaker form of localization (\emph{Anderson} localization) where the decay rate of the eigenfunctions may not be uniform (\S7 of \cite{Rio}), is enough. It is challenging to rigorously demonstrate that such localization occurs, even for specific models of random potentials in Euclidean space, see e.g.\ \cite{Wang}.

On the other hand, it should not be the case that disorder causes an entire Landau band to become spanned by localized eigenstates, otherwise the Hall conductance would remain constant across the whole Landau band instead of jumping in value between the plateaux \cite{AG,AvSS,BES,Germinet,Halperin,PSB}. 
Our Theorem \ref{thm:main} gives a very quick conceptual proof of delocalization in the Wannier sense (Eq.\ \eqref{eqn:UL.decay}), for much more general topological insulator systems and geometries than the quantum Hall effect on the Euclidean plane. We should mention that the stronger \emph{Anderson} delocalization of Landau bands is known in certain Euclidean space models \cite{BES,Germinet} using different analytic techniques. Whether our coarse geometric techniques can be combined with the latter is under study.

Finally, for a possible generalization of Wannier non-localizability when no spectral gaps are known to be available even in the $V=0$ clean limit (e.g.\ the Landau bands broaden and overlap due to variations in the Riemannian curvature of $X$ and/or magnetic field), see \cite{KLT} for the notion of a \emph{delocalized coarse index}.

\section*{Acknowledgements}
M.L.\ acknowledges support from SFB 1085 ``Higher invariants'' of the DFG.
G.C.T.\ thanks G.M.\ Graf, E.\ Prodan, and Y.\ Kubota for helpful correspondence, and acknowledges support from Australian Research Council DP200100729.

\section*{Author Declarations}
\subsubsection*{Conflict of Interest}
The authors have no conflicts to disclose.
\subsubsection*{Data Availability}
Data sharing is not applicable to this article as no new data were created or analyzed in this study.

\end{document}